\documentclass[envcountsame]{llncs}

\usepackage{amssymb,amsmath,amsfonts}
\usepackage{rotating}
\usepackage{tikz}
\usepackage[all]{xy}
\usepackage{fancyvrb}
\usepackage{refcount}
\usepackage{url}

\newtheorem{assumption}[theorem]{Assumption}

\newcommand{\ol}{\overline} 
\newcommand{\olC}{\overline{C}} 
\newcommand{\poim}{\textrm{POIM} } % changer ? % xspace ?
\newcommand{\fpoim}{\mathit{PoIm} } % changer ? 
\newcommand{\rel}{\mathit{Rel}} 
\newcommand{\gr}{\mathit{Gr}} 
\newcommand{\low}{\mathit{low}} 
\newcommand{\high}{\mathit{high}} 
\newcommand{\Match}{\mathcal{M}\!\mathit{atch}} 

\newcommand{\rhoP}{P} %?

\newcommand{\Tr}{\mathit{Tr}} 

\newcommand{\iri}{\mathit{Iri}}
\newcommand{\lit}{\mathit{Lit}}  % warning: keep ``L'' for L-K-R
\newcommand{\V}{V}
\newcommand{\B}{B}
\newcommand{\I}{I}  % Iri + Lit 
\newcommand{\IB}{{IB}}
\newcommand{\IBV}{{IBV}} 
 
\newcommand{\IV}{{IV}}

\newcommand{\tgr}[2]{\mathcal{G}_{#1}(#2)} 
\newcommand{\DGr}[1]{\mathcal{D}_{#1}} 
\newcommand{\QGr}[1]{\mathcal{Q}_{#1}} 

\newcommand{\parto}{\xymatrix@C=1pc{\ar@{^>}[r] &}} 
\newcommand{\upto}[1]{\stackrel{#1}{\to}} 
\newcommand{\from}{\leftarrow} 
\newcommand{\upfrom}[1]{\stackrel{#1}{\from}}

\newcommand{\empw}[1]{#1}
\newcommand{\doc}[1]{\textit{#1}} 
\newcommand{\blank}[1]{$\_\colon\!$#1}

\newcommand{\hide}[1]{}
\newcommand{\masque}[1]{}

\begin{document}

\title{On foundational aspects of RDF and SPARQL \\
       (Revised Version)} %TODO Please add

% \author{Dominique Duval\inst{1}, Rachid
%   Echahed\inst{2} and Fr{\'e}d{\'e}ric Prost\inst{2}}
% \institute{
% LJK, CNRS and Univ. Grenoble Alpes, Grenoble, France
% %\\ \email{Dominique.Duval@imag.fr}
% \and LIG, CNRS and Univ.  Grenoble Alpes, Grenoble, France
% %\\  \email{Rachid.Echahed@imag.fr},  \email{Frederic.Prost@imag.fr}
% }

\author{Dominique Duval, Rachid
  Echahed and Fr{\'e}d{\'e}ric Prost}
\institute{
 CNRS and Univ. Grenoble Alpes, Grenoble, France
}

\maketitle

%\keywords{RDF graphs, SPARQL queries, Category theory, Graph transformations}%TODO mandatory; please add comma-separated list of keywords

\begin{abstract}
  We consider the recommendations of the World Wide Web Consortium
  (W3C) about RDF framework and its associated query language
  SPARQL. We propose a new formal framework based on category theory
  which provides clear and concise formal definitions of the main
  basic features of RDF and SPARQL. We define RDF graphs as well as
  SPARQL basic graph patterns as objects of some nested
  categories. This allows one to clarify, in particular, the role of
  blank nodes.  Furthermore, we consider basic SPARQL CONSTRUCT and
  SELECT queries and formalize their operational semantics following a
  novel algebraic graph transformation approach called POIM.
\end{abstract}

%\input{1-intro}

%-------------------------------------------------------------------------------
%-------------------------------------------------------------------------------
\section{Introduction}

Graph databases are becoming a very influential technology in our
society. Mastering the languages involved in the
encoding of data or the formulation of queries is a necessity to
elaborate robust data management systems.

In this paper, we consider the most recent recommendations of the
World Wide Web Consortium (W3C) about the Resource Description
Framework (RDF) \cite{rdf} and the associated query language SPARQL
\cite{sparql} and propose a mathematical semantics of a kernel of
these formalisms.

The key data structure in RDF is the structure of \empw{RDF graph}.
In \cite[Section~3]{rdf}, an RDF graph is defined as a set of \empw{RDF
triples}, where an RDF triple has the form
$(\empw{subject},\empw{predicate},\empw{object})$. The
\empw{subject} is either an \empw{IRI} (Internationalized Resource
Identifier) or a \empw{blank node}, the \empw{predicate} is an IRI and
the \empw{object} is either an IRI, a literal (denoting a value such as
a string, a number or a date) or a  \empw{blank node}. 
Blank nodes are arbitrary elements as long as they differ from IRIs
and literals and they do not have any internal structure: they are
used for indicating the existence of a thing and the \empw{blank node
  identifiers} are \empw{locally scoped}.
For instance, the triples $(Paul, knows, blank1)$ and
$(blank2, knows, Henry)$ mean, respectively, that Paul knows someone and
someone knows Henry. Surprisingly, a triple such as $(Paul, blank3,
Henry)$ standing for ``there is some relationship between Paul and Henry''
is not allowed in RDF, but only in generalized RDF
\cite[Section~7]{rdf}. Following the theoretical point of view we
propose in this paper, there is no harm to consider blank
predicates within RDF triples. We thus consider \emph{data graphs} in a more
general setting including RDF graphs.

%%%%%%%%%%%%%%%%%%%%%%%%%%%%%%%%%%%%%%%%%%%%%%%%

The query language SPARQL for RDF databases is based on \empw{basic
  graph patterns}, which are kinds of RDF graphs with variables
\cite[Section~2]{sparql}. In this paper, we consider \emph{query
  graphs} which generalize basic graph patterns by allowing blanks to
be predicates.  The SPARQL query processor searches for triples within
a given RDF database which \empw{match} the triple patterns
in the given basic graph pattern, and returns a multiset of solutions
or an RDF graph. 
Considering basic graph patterns, one may wonder what is the
difference between variables and blank nodes. SPARQL
specifications in \cite[Section~4.1.4]{sparql}
suggest similarities between them,  
whereas the opposite is made in  \cite[Section~16.2]{sparql}.
In the formalization of 
SPARQL we propose, blank nodes
and variables are clearly distinguished by their respective roles in
the definition of morphisms.

In the SPARQL recommendation \cite{sparql}, the SELECT query form is described
lengthily. This query form can be compared to the SELECT
query form of SQL, which returns a multiset of solutions. 
In contrast, the CONSTRUCT query form  returns an RDF graph.
The latter is described very shortly in \cite[Section~16.2]{sparql}.
Following our formalization, the CONSTRUCT  query form is  more
fundamental than the SELECT query form. Actually, we start by  proposing an
operational semantics for CONSTRUCT queries based on a new approach of
algebraic graph transformations which we call \poim and we show afterward how
SELECT queries can be easily encoded as CONSTRUCT queries.

The paper is organized as follows. Section~\ref{sec:graph} defines the
objects and the morphisms of the categories of data graphs and query
graphs. Section~\ref{sec:poim} introduces the \poim algebraic
transformation: a rewrite rule is a cospan $L \rightarrow K \leftarrow R$
where $L,K$ and $R$ are basic graph patterns, and a rewrite step is made of
a pushout followed by an image factorization. 
Afterward, in Section~\ref{sec:construct} we define two different
operational semantics for CONSTRUCT queries and prove their equivalence.
We first define a \emph{high-level calculus} as a mere application of the \poim
transformation. Then we propose a \emph{low-level calculus} which is defined by
means of several applications of the \poim transformation followed by
a ``merging'' process. Both calculi implement faithfully the SPARQL
semantics for CONSTRUCT queries (Theorem~\ref{theo:construct-sparql}).
In Section~\ref{sec:select}, we show how the \poim transformation can
be used to define a novel operational semantics of the SELECT queries.
This semantics, which is faithful to SPARQL definitions
(Theorem~\ref{theo:select-sparql}), is obtained by an original
translation of each SELECT query into a CONSTRUCT query.
Concluding remarks and related work are discussed in Section~\ref{sec:discuss}.
The missing proofs are in the Appendix.

%-------------------------------------------------------------------------------
%-------------------------------------------------------------------------------
\section{Graphs of triples} 
\label{sec:graph}

The set of \empw{IRIs}, denoted $\iri$, and the set of
\empw{literals}, 
denoted $\lit$, with its usual operations, 
are defined in \cite{rdf}. 
The sets $\iri$ and $\lit$ are disjoint. 
In addition, let $\B$ be a countably infinite set,
disjoint from $\iri$ and $\lit$. 
The elements of $\B$ are called the \emph{blanks}.
According to \cite[Section~3.1]{rdf}, 
\doc{an \empw{RDF graph} is a set of RDF triples 
and an \empw{RDF triple} consists of three components:
the \empw{subject}, which is an IRI or a blank node; 
the \empw{predicate}, which is an IRI; 
and the \empw{object}, which is an IRI, a literal or a blank node. 
The set of \empw{nodes} of an RDF graph is the set of subjects and objects
of triples in the graph.}
Using set-theoretic notations, this can be expressed as follows: 
let $\Tr= (\iri\cup\B) \times \iri \times (\iri\cup\lit\cup\B)$, 
then an RDF triple is an element of $\Tr$ and an RDF graph is a subset of $\Tr$.
Let us also consider the following extension of RDF \cite[Section~7]{rdf}:
\doc{
  A \empw{generalized RDF triple} is a triple having a subject, a predicate,
  and object, where each can be an IRI, a blank node or a literal.
  A \empw{generalized RDF graph} is a set of generalized RDF triples.
}
Let $\I=\iri\cup\lit$, so that 
a generalized RDF triple is an element of $(\I\cup\B)^3$
and a generalized RDF graph is a subset of $(\I\cup\B)^3$.

Let $\V$ be a countably infinite set disjoint from $\iri$, $\lit$ and $\B$. 
The elements of $\V$ are called the \emph{variables}.
According to \cite[Section~2]{sparql}
\doc{a set of triple patterns is called a basic graph pattern.
Triple patterns are like RDF triples except that each of the subject,
predicate and object may be a variable}.
Let $\Tr_V= (\iri\cup\B\cup\V) \times (\iri\cup\V) \times
(\iri\cup\lit\cup\B\cup\V)$, 
then a triple pattern is an element of $\Tr_V$
and a basic graph pattern is a subset of $\Tr_V$.
Since $\Tr_V$ is a subset of $(\I\cup\B\cup\V)^3$,
each basic graph pattern is a subset of $(\I\cup\B\cup\V)^3$.

RDF graphs and basic graph patterns are generalized 
in Definition~\ref{def:graph-data-query} as \emph{data graphs} and
\emph{query graphs} respectively,
both relying on Definition~\ref{def:graph-graph}.  

\begin{definition}
\label{def:graph-graph}
For each set $A$, the \emph{triples on $A$} are the elements of $A^3$. 
For each triple $t=(s,p,o)$ on $A$ 
the elements $s$, $p$ and $o$ of $A$ are called respectively
the \emph{subject}, the \emph{predicate} and the \emph{object} of $t$.
A \emph{graph on $A$} is a set of triples on $A$, i.e. a subset of $A^3$. 
For each graph $T$ on $A$, the subset of $A$ made of the
subjects, predicates and objects of $T$ is called the set of \emph{attributes}
of $T$ and is denoted $|T|$; it follows that $T$ is a subset of $|T|^3$. 
Let $T$ and $T'$ be two graphs on $A$.
A \emph{morphism} $a:T\to T'$ is a map such that there is a map $M:|T|\to |T'|$
such that $a$ is the restriction of $M^3$ to $T$. 
Then $M$ is uniquely determined by $a$, it is denoted $|a|$. 
This yields the \emph{category of graphs on $A$}, denoted
$\tgr{}{A}$. 
We say that a morphism $a:T\to T'$ of graphs on $A$ \emph{fixes} a subset $C$ of $A$ 
if $|a|(x)=x$ for each $x$ in $|T| \cap C$.
For each subset $C$ of $A$, the subcategory of $\tgr{}{A}$ made of
the graphs on $A$ with the morphisms fixing $C$ is denoted $\tgr{C}{A}$.
\end{definition}

Thus, by mapping $a$ to $|a|$ we get a one-to-one correspondence
between the morphisms $a:T\to T'$ of graphs on $A$ 
and the maps $M:|T|\to |T'|$ such that $M^3(T)\subseteq T'$.
An isomorphism (i.e., an invertible morphism) in $\tgr{}{A}$ 
is a morphism $a:T\to T'$ of graphs on $A$ such that $|a|:|T|\to |T'|$
is a bijection and $a(T)=T'$. 
A morphism $a$ fixing $C$ is determined by the restriction
of the map $|a|$ to $|T|\cap \ol{C}$, where $\ol{C} = A\setminus C$.
An isomorphism $a$ in $\tgr{C}{A}$
is a morphism $a:T\to T'$ of graphs on $A$ such that
$|a|$ is the identity on $|T| \cap C$
and a bijection between $|T|\cap \ol{C}$ and $|T'|\cap \ol{C}$ 
and $a(T)=T'$. 
The notions of \emph{inclusion}, \emph{subgraph}, \emph{image} and \emph{union}
for graphs on $A$ are 
defined as inclusion, subset, image and union for subsets of $A^3$.

\begin{definition}
\label{def:graph-data-query}
Let $\I$, $\B$ and $\V$ be three pairwise distinct countably infinite sets,
called respectively the sets of \emph{resource identifiers},
\emph{blanks} and \emph{variables}. 
Let $\IB=\I\cup\B$, $\IV=\I\cup\V$ and $\IBV=\I\cup\B\cup\V$.
The \emph{category of data graphs} is $\DGr{}=\tgr{}{\IB}$ and 
for each subset $C$ of $\IB$ the category of data graphs \emph{fixing} $C$
is the subcategory $\DGr{C}=\tgr{C}{\IB}$ of $\DGr{}$. 
The \emph{category of query graphs} is $\QGr{}=\tgr{}{\IBV}$ and 
for each subset $C$ of $\IBV$ the category of query graphs \emph{fixing} $C$
is the subcategory $\QGr{C}=\tgr{C}{\IBV}$ of $\QGr{}$. 
\end{definition}

Thus, when $\I=\iri\cup\lit$, 
the RDF graphs are the data graphs where 
only nodes can be blanks and only nodes that are not subjects can be literals,
and the \empw{RDF terms} of an RDF graph are its attributes
when it is seen as a data graph. 
Then the \empw{isomorphisms of RDF graphs}, as defined in
\cite[Section~3.6.]{rdf}, are the isomorphisms 
in the category $\DGr{\I}$ of data graphs fixing $\I$: 
indeed, two data graphs $G_1$ and $G_2$ are isomorphic in $\DGr{\I}$
if and only if they differ only by the names of their blanks. 
For each data graph $T$, let $|T|_\I=|T|\cap\I$ and $|T|_\B=|T|\cap\B$,
so that $|T|$ is the disjoint union of $|T|_\I$ and $|T|_\B$.
Similarly, the basic graph patterns of SPARQL are the query graphs where 
only nodes can be blanks and only nodes that are not subjects can be literals.
For each query graph $T$, let $|T|_\I=|T|\cap\I$, $|T|_\B=|T|\cap\B$ and
$|T|_\V=|T|\cap\V$,
so that $|T|$ is the disjoint union of $|T|_\I$, $|T|_\B$ and $|T|_\V$.

Morphisms of graphs can be used, for instance, 
for substituting the variables of a query graph 
(Definition~\ref{def:graph-match})
or for interpreting a data graph in a universe of discourse
(Definition~\ref{def:rdf-interpretation}).

\begin{definition}
\label{def:graph-match}
A \emph{match} from a query graph $L$ to a data graph $G$ 
is a morphism of query graphs from $L$ to $G$ which fixes $\I$. 
The set of matches from $L$ to $G$ is denoted $\Match(L,G)$
and the set of all matches from $L$ to any data graph is denoted $\Match(L)$.
\end{definition}

Thus, a match fixes each resource identifier   
and it maps a variable or a blank to a resource identifier or a blank.

The interpretations of an RDF graph are also kinds of morphisms, 
see Definition~\ref{def:rdf-interpretation}. 
Note that this will not be used later in this paper. 
We define an interpretation of
a data graph $G$ in a universe of discourse $U$ by generalizing
the definition of a morphism, according to \cite[Section~1.2.]{rdf}: 
\doc{
Any IRI or literal denotes something in the world 
(the ``\empw{universe of discourse}'').
These things are called \empw{resources}. 
The predicate itself is an IRI and denotes a \empw{property}, that is,
a resource that can be thought of as a binary relation.
}
Recall that the binary relations on a set $R$ are the subsets of $R^2$.
It can happen that a binary relation on $R$ is itself an element of $R$. 

\begin{definition}
\label{def:rdf-interpretation}
Given a set $R$ and a subset $P$ of $R^2$ made of binary relations on $R$, 
let $U$ be the set of triples $(s,p,o)$ in $R^3$ such that $p\in P$
and $(s,o)\in p$.
The \emph{universe of discourse} with $R$ as set of \emph{resources}
and $P$ as set of \emph{properties} is the graph $U$ on $R$.  
Given a universe of discourse $U$ on a set $R$ and a map $M_\I:\I \to R$, 
an \emph{interpretation} of a data graph $G$ is a map $i:G\to U$
such that $i=M^3$ for a map $M:|G|\to |U|$ which extends $M_\I$.
\end{definition}

In this paper, we consider categories $\DGr{C}$ and $\QGr{C}$ for various
subsets $C$ of $\IB$ and $\IBV$ respectively.
It will always be the case that $C$ contains $\I$,
so that we can say that resource identifiers have a ``global scope''. 
In contrast, blanks have a ``local scope'': 
in the basic part of RDF and SPARQL considered in this paper,
the scope of a blank node is restricted to one data graph or one query graph.
The note about \emph{blank node identifiers} in \cite[Section~3.4]{rdf} distinguishes
two kinds of syntaxes for RDF:
an abstract syntax where blank nodes do not have identifiers
and concrete syntaxes where blank nodes have identifiers.
In our approach a blank \emph{is} an attribute, which corresponds to a
concrete syntax, 
and the abstract syntax is obtained by considering data graphs
as objects of the category $\DGr{\I}$ up to isomorphism,
so that any blank node can be changed for a new blank node if needed.

\begin{example}
  \label{ex::graph-categories}
In all examples we use the following prefixes
(\texttt{@prefix} for data and \texttt{PREFIX} for queries): 

\begin{Verbatim}[frame=single,label=Prefixes,fontsize=\scriptsize]
@prefix  foaf:  <http://xmlns.com/foaf/0.1/>.
PREFIX foaf:    <http://xmlns.com/foaf/0.1/>
PREFIX vcard:   <http://www.w3.org/2001/vcard-rdf/3.0#>
\end{Verbatim}

\noindent
Consider two RDF graphs $G_1$, $G_2$ as follows.
They are isomorphic in $\DGr{\I}$
but not in $\DGr{\IB}$ because blanks are swapped.  

\medskip 
\noindent
\begin{minipage}{2.2in}
\begin{Verbatim}[frame=single,label=$G_1$,fontsize=\scriptsize]
<http://example.org/Al> foaf:knows _:b. 
_:c foaf:knows <http://example.org/Bob>. 
\end{Verbatim}
	\end{minipage} \hfill
	\begin{minipage}{2.2in}
\begin{Verbatim}[frame=single,label=$G_2$,fontsize=\scriptsize]
<http://example.org/Al> foaf:knows _:c. 
_:b foaf:knows <http://example.org/Bob>. 
\end{Verbatim}
	\end{minipage}

\medskip 
\noindent
Now consider basic graph patterns $G_3$ to $G_8$.
They are pairwise non-isomorphic in $\QGr{\IBV}$ because they are pairwise
distinct.
In $\QGr{\IV}$ only $G_7$ and $G_8$ are isomorphic. 
In $\QGr{\I}$ these query graphs belong to two different isomorphism classes:
on one side $G_3$ and $G_4$ are isomorphic
and on the other side $G_5$, $G_6$, $G_7$ and $G_8$ are isomorphic.

\medskip 
\noindent
	\begin{minipage}{2.2in}
\begin{Verbatim}[frame=single,label=$G_3$,fontsize=\scriptsize]
<http://example.org/Al> foaf:knows _:b. 
_:b foaf:knows <http://example.org/Bob>. 
\end{Verbatim}
	\end{minipage} \hfill
	\begin{minipage}{2.2in}
\begin{Verbatim}[frame=single,label=$G_4$,fontsize=\scriptsize]
<http://example.org/Al> foaf:knows ?x. 
?x foaf:knows <http://example.org/Bob>. 
\end{Verbatim}
	\end{minipage}

\medskip	
\noindent
	\begin{minipage}{2.2in}
\begin{Verbatim}[frame=single,label=$G_5$,fontsize=\scriptsize]
<http://example.org/Al> foaf:knows _:b. 
_:c foaf:knows <http://example.org/Bob>. 
\end{Verbatim}
	\end{minipage} \hfill
	\begin{minipage}{2.2in}
\begin{Verbatim}[frame=single,label=$G_6$,fontsize=\scriptsize]
<http://example.org/Al> foaf:knows ?x. 
?y foaf:knows <http://example.org/Bob>. 
\end{Verbatim}
	\end{minipage}

\medskip	
\noindent
	\begin{minipage}{2.2in}
\begin{Verbatim}[frame=single,label=$G_7$,fontsize=\scriptsize]
<http://example.org/Al> foaf:knows ?x. 
_:b foaf:knows <http://example.org/Bob>. 
\end{Verbatim}
	\end{minipage} \hfill
	\begin{minipage}{2.2in}
\begin{Verbatim}[frame=single,label=$G_8$,fontsize=\scriptsize]
<http://example.org/Al> foaf:knows ?x. 
_:c foaf:knows <http://example.org/Bob>. 
\end{Verbatim}
	\end{minipage}

\end{example}	

\begin{assumption}
\label{ass:graph}
   {F}rom now on $A$ is a set, $C$ is a subset of $A$,
    $\olC=A\setminus C$ is the complement of $C$ in $A$,
    and it is assumed that both $C$ and $\olC$ are countably infinite. 
\end{assumption}

\begin{remark}
\label{rem:graph-iso}
Since $\olC$ is countably infinite, when dealing with a finite number of finite graphs
on $A$ it is always possible to find a \emph{new attribute outside $C$},
i.e., an element of $\olC$ that is not an attribute of any of the given graphs.
We will use repeatedly the following consequence of this fact:
\smallskip 
\\ \textsl{Given a graph $T$ on $A$, if any attribute of $T$ in $\olC$
is replaced by any new element of $\olC$ the result is 
a graph $T'$ on $A$ that is isomorphic to $T$ in $\tgr{C}{A}$.
Such a $T'$ exists when $T$ is finite.}
\end{remark}

Now let us focus on some kinds of colimits of graphs on $A$:
coproducts in Proposition~\ref{prop:graph-coprod} and 
pushouts in Proposition~\ref{prop:graph-po}.
Recall that colimits in any category are defined up to isomorphism
in this category. 
  
\begin{proposition}  
\label{prop:graph-coprod}
Given graphs $T_1,...,T_k$ on $A$
such that $|T_i|\cap|T_j|\subseteq C$ for each $i\ne j$,
the union $T_1 \cup ...\cup T_k$ is a coproduct of $T_1,...,T_k$
in $\tgr{C}{A}$.
\end{proposition}

By Remark~\ref{rem:graph-iso} it follows that
if $T_1,...,T_k$ are any finite graphs on $A$ there are graphs $T'_1,...,T'_k$
on $A$ such that $T'_i$ is isomorphic to $T_i$ in $\tgr{C}{A}$ for each $i$ 
and $|T'_i|\cap|T'_j|\subseteq C$ for each $i\ne j$,
so that the union $T'_1 \cup ...\cup T'_k$ is a coproduct of $T_1,...,T_k$
in $\tgr{C}{A}$.

\begin{proposition}  
\label{prop:graph-po}
Let $l:L\to K$ and $m:L\to G$ be morphisms of graphs on $A$
such that $K$ is finite, $l$ is an inclusion and $m$ fixes $C$. 
Let us assume that $|K| \cap |G| \subseteq C$
(this is always possible up to isomorphism in $\tgr{C}{A}$,
by Remark~\ref{rem:graph-iso}).
Let $N:|K|\to |G|\cup|K\setminus L|$ be such that $N(x)=|m|(x)$ for $x\in |L|$ 
and $N(x)=x$ otherwise. 
Let $D=G\cup N^3(K)$, let $n:K\to D$ be the restriction of $N^3$
and $g:G\to D$ the inclusion.
Then $|D|=|G|\cup|K\setminus L|$ and
the square $(l,m,n,g)$ is a pushout square in $\tgr{C}{A}$.
\end{proposition}

Thus, $D$ is a kind of ``union of $G$ and $K$ over $L$'', however 
in general it is \emph{not} the case that $D$ is the union of $G$ and
$K\setminus L$. 
It is the case that $D=G\cup D_2$ where $D_2=N^3(K\setminus L)$
but $N^3$ is not the identity on $K\setminus L$,
and moreover $G$ and $D_2$ are not disjoint in general.

%\input{3-poim}

%-------------------------------------------------------------------------------
%-------------------------------------------------------------------------------
\section{The \poim transformation} 
\label{sec:poim}

A SPARQL query like 
``CONSTRUCT \{$R$\} WHERE \{$L$\}'' is called \emph{basic} 
when both $R$ and $L$ are basic graph patterns.
In such a query, variables with the same name in $L$ and $R$ denote the same
RDF term, whereas it is not the case for blank nodes.
The statement ``\doc{blank nodes in graph patterns act as variables}''
in \cite[Section~4.1.4]{sparql} holds for $L$, whereas blank nodes in $R$
give rise to new blank nodes in the result of the query as in
Examples~\ref{ex:poim-one-B-run} and~\ref{ex:construct-two-B-run}.
Thus, the meaning of blank nodes in $L$ is unrelated to the
meaning of blank nodes in $R$, 
and in both $L$ and $R$ each blank can be replaced by a new blank.

We generalize this situation in Definition~\ref{def:construct-rule}
by allowing any data graphs for $L$ and $R$ up to isomorphism in $\QGr{\IV}$:
the resource identifiers and the variables in $L$ and $R$ are fixed 
but each blank can be replaced by a new blank. 
Thus, without loss of generality, we can assume that 
$|L|_\B\cap|R|_\B=\emptyset$.
Under this assumption, the set of triples $K=L\cup R$  
with the inclusions of $L$ and $R$ in $K$ 
is a coproduct of $L$ and $R$ in the category $\QGr{\IV}$. 
We also assume that each variable in $R$ occurs in $L$,
so that every substitution for the variables in $L$ provides a 
substitution for the variables in $R$.
This assumption $|R|_\V\subseteq|L|_\V$ is equivalent to $|K|_\V=|L|_\V$.

\begin{definition}
\label{def:construct-rule}
A \emph{basic construct query} is a pair of finite query graphs $(L,R)$  
such that $|L|_\B\cap|R|_\B=\emptyset$ and $|R|_\V\subseteq|L|_\V$, 
up to isomorphism in the category $\QGr{\IV}$. 
The \emph{transformation rule} of a basic construct query $(L,R)$ is the
cospan $\rhoP_{L,R}=(L\upto{l} K\upfrom{r} R)$ 
where $K=L\cup R$ and $l$ and $r$ are the inclusions.
Its \emph{left-hand side} is $L$ and its \emph{right-hand side} is $R$.

$$ 
\rhoP_{L,R} \;=\; \xymatrix@C=4pc{
  L \ar[r]^(.35){l}_(.35){\subseteq} & K=L\cup R & \ar[l]_(.35){r}^(.35){\supseteq} R} 
$$ 
\end{definition}

\begin{example}
\label{ex:poim-one-V}

Consider the following SPARQL CONSTRUCT query:

\begin{Verbatim}[frame=single,label=Query,fontsize=\scriptsize]
 CONSTRUCT { ?x vcard:FN ?name } WHERE { ?x foaf:name ?name }
\end{Verbatim}

\noindent
In the corresponding transformation rule $L\upto{l} K\upfrom{r} R$
there are no blanks in $L$ nor in $R$,
thus the transformation rule is $L\upto{l} K\upfrom{r} R$
where $l$ and $r$ are the inclusions of $L$ and $R$ in $K=L\cup R$.

\medskip 
\noindent
\begin{minipage}{1.3in}
\begin{Verbatim}[frame=single,label=$L$,fontsize=\scriptsize] 
 ?x foaf:name ?name . 
\end{Verbatim} 
\end{minipage}
\hfill 
  $\upto{l}$
\hfill 
\begin{minipage}{1.3in}	
\begin{Verbatim}[frame=single,label=$K$,fontsize=\scriptsize] 
 ?x foaf:name ?name ;
    vcard:FN ?name .
\end{Verbatim}
\end{minipage}
\hfill 
  $\upfrom{r}$
\hfill 
\begin{minipage}{1.3in}	
\begin{Verbatim}[frame=single,label=$R$,fontsize=\scriptsize] 
 ?x vcard:FN ?name .
\end{Verbatim} 
\end{minipage}

\end{example} 

\begin{example}
\label{ex:poim-one-B}

Now the SPARQL CONSTRUCT query from Example~\ref{ex:poim-one-V}
is modified by replacing both occurrences of the variable \mbox{{\tt ?x}}
by the blank node \mbox{{\tt \blank{x}}}: 

\begin{Verbatim}[frame=single,label=Query,fontsize=\scriptsize]
 CONSTRUCT { _:x vcard:FN ?name } WHERE { _:x foaf:name ?name }
\end{Verbatim}

\noindent 
In the corresponding transformation rule
one blank has been modified so as to ensure that $|L|_\B\cap|R|_\B$ is empty:

\medskip 
\noindent
\begin{minipage}{1.3in}
\begin{Verbatim}[frame=single,label=$L$,fontsize=\scriptsize] 
 _:x foaf:name ?name . 
\end{Verbatim} 
\end{minipage}
\hfill 
  $\upto{l}$
\hfill 
\begin{minipage}{1.3in}	
\begin{Verbatim}[frame=single,label=$K$,fontsize=\scriptsize] 
 _:x foaf:name ?name .
 _:y vcard:FN ?name .
\end{Verbatim}
\end{minipage}
\hfill 
  $\upfrom{r}$
\hfill 
\begin{minipage}{1.3in}	
\begin{Verbatim}[frame=single,label=$R$,fontsize=\scriptsize] 
 _:y vcard:FN ?name .
\end{Verbatim} 
\end{minipage}

\end{example} 

When a basic SPARQL query ``CONSTRUCT \{$R$\} WHERE \{$L$\}'' is run
against an RDF graph $G$, and when there is precisely one match of $L$
into $G$, the result of the query is an RDF graph $H$ obtained by
substituting the variables in $R$. This
substitution can be seen as a match of $R$ into $H$.
We claim that the process of building $H$ with this match of $R$ into $H$
from the match of $L$ into $G$ can be seen as a two-step process involving an
intermediate match of $K$ in an RDF graph $D$.
The definition of this process relies on an algebraic construction
that we call the \emph{\poim transformation}: 
PO for \emph{pushout} and IM for \emph{image}
(Definition~\ref{def:construct-poim}).
The \poim transformation is related 
to a large family of algebraic graph transformations based on pushouts,
like the SPO (Simple Pushout) \cite{DBLP:conf/gg/EhrigHKLRWC97}, 
DPO (Double Pushout) \cite{CorradiniMREHL97}
or SqPO (Sesqui-Pushout) \cite{CorradiniHHK06}. 

Given a basic construct query $(L,R)$ and its transformation rule
$L\upto{l} K\upfrom{r} R$, 
the \poim transformation is defined as a map from the matches of $L$
to the matches of $R$, in two steps:
first from the matches of $L$ to the matches of $K$,
then from the matches of $K$ to the matches of $R$.  
Given an inclusion $l:L\to K$ in $\QGr{\I}$, 
the \emph{cobase change along $l$} is the 
map $l_*:\Match(L) \to \Match(K)$ that maps
each $m:L\to G$ to $l_*(m):K\to D$ defined from the pushout of $l$ and $m$
in $\QGr{\I}$, as described in Proposition~\ref{prop:graph-po}. 
Note that $D$ is a data graph because of the assumption $|K|_\V=|L|_\V$.
Given an inclusion $r:R\to K$ in $\QGr{\I}$, 
the \emph{image factorization along $r$} is the 
map $r^+:\Match(K) \to \Match(R)$ that maps
each $n:K\to D$ to $r^+(n):R\to H$ where $H$ is the image 
of $R$ in $D$ and $r^+(n)$ is the restriction of $n$ 
and $h:H\to D$ is the inclusion.
This leads to Definition~\ref{def:construct-poim}
and Proposition~\ref{prop:construct-poim}.

\begin{definition}
\label{def:construct-poim}
Let $(L,R)$ be a basic construct query and $L\upto{l} K\upfrom{r} R$
its transformation rule.
The \emph{\poim transformation map} of $(L,R)$ is the map 
$$ \fpoim_{L,R} = r^+\circ l_*:\Match(L) \to \Match(R) $$
composed of the cobase change map $l_*$ and the image factorization map $r^+$.
The \emph{result} of applying $\fpoim_{L,R}$ to a match $m:L\to G$
is the match $\fpoim_{L,R}(m):R\to H$ or simply the query graph $H$. 
\end{definition}
\begin{equation}
  \label{diag:poim}
\xymatrix@R=2.5pc@C=6pc{
  \ar@{}[rd]|{(PO)} L \ar[r]^{l} \ar[d]_(.4){m} &
  K \ar@{.>}[d]_(.4){l_*(m)}|(.4){=}^(.4){n} \ar@{}[rd]|{(IM)} & 
  R \ar[l]_{r} \ar@{.>}[d]_(.4){r^+(n)}|(.4){=}^(.4){p}|(.6){=}^(.6){\fpoim_{L,R}(m)} \\
  G \ar@{.>}[r]^{g} & D & H \ar@{.>}[l]_{h} \\ 
}
\end{equation}

Note that the result $H$ is defined only up to isomorphism in $\QGr{\I}$,
which means that the blanks in $H$ can be modified  
(as long as this modification does not identify any of them). 

\begin{proposition}
\label{prop:construct-poim}
Let $(L,R)$ be a basic construct query and $m:L\to G$ a match. 
Let $P:|R|\to A$ be defined by $P(x)=|m|(x)$ for $x\in|R|_\V$ and $P(x)=x$
otherwise.
Then, up to isomorphism in $\QGr{\I}$, 
the result of applying $\fpoim_{L,R}$ to $m$ is $p:R\to H$
where $H=P^3(R)$ and $p$ is the restriction of $P^3$. 
\end{proposition}

\begin{remark}
  \label{rem:poim-category}
  Each set $\Match(X)$ can be seen as a coslice category, then the maps
  $r^+$ and $l_*$ can be seen as functors: this could be useful
  when extending this paper to additional features of SPARQL. 
\end{remark}

\begin{example}
  \label{ex:poim-one-V-run}
  
Consider the SPARQL CONSTRUCT query from Example~\ref{ex:poim-one-V}: 

\begin{Verbatim}[frame=single,label=Query,fontsize=\scriptsize]
 CONSTRUCT { ?x vcard:FN ?name } WHERE { ?x foaf:name ?name }
\end{Verbatim}

\noindent
and let us run this query against the RDF graph $G$:

\begin{Verbatim}[frame=single,label=$G$,fontsize=\scriptsize] 
 ex:a foaf:name "Alice" ; foaf:nick "Lissie" .
\end{Verbatim}

\noindent
There is a single match $m$, it is such that 
$ m({\tt ?x}) = \mbox{{\tt ex:a}} $ and
$ m(\mbox{{\tt ?name}}) = \mbox{{\tt "Alice"}} $.
The \poim transformation produces successively the following
data graphs $D$ and $H$, where $H$ is the query result:

\medskip 
\noindent
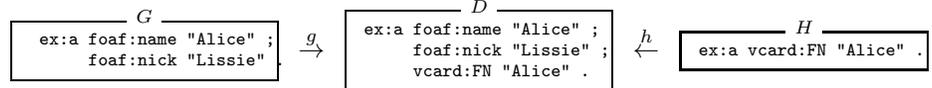
\begin{minipage}{1.4in}
\begin{Verbatim}[frame=single,label=$G$,fontsize=\scriptsize] 
  ex:a foaf:name "Alice" ;
       foaf:nick "Lissie" .
\end{Verbatim} 
\end{minipage}
\hfill 
  $\upto{g}$
\hfill 
\begin{minipage}{1.4in}	
\begin{Verbatim}[frame=single,label=$D$,fontsize=\scriptsize] 
 ex:a foaf:name "Alice" ;
      foaf:nick "Lissie" ; 
      vcard:FN "Alice" . 
\end{Verbatim}
\end{minipage}
\hfill 
  $\upfrom{h}$
\hfill 
\begin{minipage}{1.3in}	
\begin{Verbatim}[frame=single,label=$H$,fontsize=\scriptsize] 
 ex:a vcard:FN "Alice" . 
\end{Verbatim} 
\end{minipage}

\end{example} 

\begin{example}
\label{ex:poim-one-B-run}
  
Now consider the SPARQL CONSTRUCT query from Example~\ref{ex:poim-one-B}:  

\begin{Verbatim}[frame=single,label=Query,fontsize=\scriptsize]
 CONSTRUCT { _:x vcard:FN ?name } WHERE { _:x foaf:name ?name }
\end{Verbatim}

\noindent
Let us run this query against the RDF graph $G$ from
Example~\ref{ex:poim-one-V-run}. 
There is a single match $m$, it is such that 
$ m({\tt \_:x}) = \mbox{{\tt ex:a}} $ and
$ m(\mbox{{\tt ?name}}) = \mbox{{\tt "Alice"}} $.
The \poim transformation produces successively the following
data graphs $D$ and $H$, where $H$ is the query result:

\medskip 
\noindent
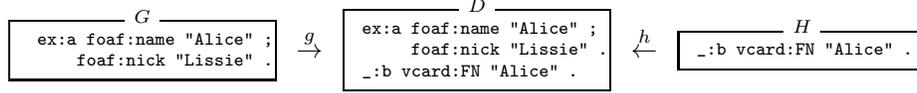
\begin{minipage}{1.4in}
\begin{Verbatim}[frame=single,label=$G$,fontsize=\scriptsize] 
  ex:a foaf:name "Alice" ;
      foaf:nick "Lissie" .
\end{Verbatim} 
\end{minipage}
\hfill 
  $\upto{g}$
\hfill 
\begin{minipage}{1.4in}	
\begin{Verbatim}[frame=single,label=$D$,fontsize=\scriptsize] 
 ex:a foaf:name "Alice" ;
      foaf:nick "Lissie" . 
 _:b vcard:FN "Alice" . 
\end{Verbatim}
\end{minipage}
\hfill 
  $\upfrom{h}$
\hfill 
\begin{minipage}{1.3in}	
\begin{Verbatim}[frame=single,label=$H$,fontsize=\scriptsize] 
 _:b vcard:FN "Alice" . 
\end{Verbatim} 
\end{minipage}

\end{example} 

%\input{4-construct}

%-------------------------------------------------------------------------------
%-------------------------------------------------------------------------------
\section{Running basic construct queries} 
\label{sec:construct}

In Section~\ref{sec:poim}, we defined the \poim transformation and
we applied it to run a basic construct query $(L,R)$ against a data graph $G$, 
under the assumption that there is exactly one match
from $L$ to $G$. Now we define two different calculi for 
running a basic construct query against a data graph $G$
without any assumption on the number of matches.
The \emph{high-level calculus} (Definition~\ref{def:construct-high})
is one ``large'' application of the \poim transformation. 
The \emph{low-level calculus} (Definition~\ref{def:construct-low})
consists of several ``small'' applications of the \poim transformation
followed by a ``merging'' process.
In Propositions~\ref{prop:construct-high} and~\ref{prop:construct-low}
we prove that both calculi return the same result.
This result coincides (up to the renaming of the blanks) 
with the result returned by SPARQL when
$L$ and $R$ are basic graph patterns and $G$ is an RDF graph
(Theorem~\ref{theo:construct-sparql}). 

\begin{definition}
\label{def:construct-result}
Let $(L,R)$ be a basic construct query and $G$ a data graph.
Assume (without loss of generality) that $|G|_\B\cap |L|_\B=\emptyset$
and $|G|_\B\cap |R|_\B=\emptyset$. 
Let $m_1,...,m_k$ be the matches from $L$ to $G$.
For each $i=1,...,k$ let $H_i$ be the data graph obtained from $R$
by replacing each variable $x$ in $R$ by $m_i(x)$
and each blank in $R$ by a new blank (which means:  
a new blank for each blank in $R$ and each $i$ in $\{1,...,k\}$).
The \emph{query result} of applying the basic construct query $(L,R)$
to the data graph $G$ is the data graph $H=H_1 \cup ...\cup H_k$.
\end{definition}

% avant : answer, well-formed etc ***+++
A triple $(s,p,o)$ in $(\I\cup \B)^3$, where $\I=\iri\cup\lit$,
is \emph{well-formed} if it is an RDF triple,
in the sense that $s\in\iri\cup \B$ and $p\in\iri$. 
Thus, a data graph is an RDF graph if and only if all its triples
are well-formed.
The \emph{answer} of a SPARQL CONSTRUCT query over an RDF graph
is defined in \cite{Kostylev_et_al2015}.

\begin{theorem}
\label{theo:construct-sparql}
Let $L$ and $R$ be basic graph patterns with $|L|_\B=\emptyset$
and $|R|_\V\subseteq |L|_\V$.
Then $(L,R)$ is a basic construct query and 
the set of well-formed triples in the query result of applying $(L,R)$
to an RDF graph $G$ is isomorphic in $\DGr{\I}$ to the answer
of the SPARQL query ``CONSTRUCT $\{R\}$ WHERE $\{L\}$'' over $G$.  
\end{theorem}

\begin{example}
\label{ex:construct-two-V-run}
  
Consider the SPARQL query from Examples~\ref{ex:poim-one-V}
and~\ref{ex:poim-one-V-run}:

\begin{Verbatim}[frame=single,label=Query,fontsize=\scriptsize]
 CONSTRUCT { ?x vcard:FN ?name } WHERE { ?x foaf:name ?name }
\end{Verbatim}

\noindent
and let us run this query against the RDF graph $G$:

\begin{Verbatim}[frame=single,label=$G$,fontsize=\scriptsize] 
 ex:a foaf:name "Alice" ; foaf:nick "Lissie" .
 ex:b foaf:name "Bob" ; foaf:nick "Bobby" .
\end{Verbatim}

\noindent 
There are two matches and we get the RDF graphs $H_1$, $H_2$
and the result $H$: 

\medskip
\noindent 
\begin{minipage}{1.3in}	
\begin{Verbatim}[frame=single,label=$H_1$,fontsize=\scriptsize] 
 ex:a vcard:FN "Alice" . 
\end{Verbatim} 
\end{minipage}
\hfill 
\begin{minipage}{1.3in}	
\begin{Verbatim}[frame=single,label=$H_2$,fontsize=\scriptsize] 
 ex:b vcard:FN "Bob" . 
\end{Verbatim} 
\end{minipage}
\hfill 
\begin{minipage}{1.3in}	
\begin{Verbatim}[frame=single,label=$H$,fontsize=\scriptsize] 
 ex:a vcard:FN "Alice" . 
 ex:b vcard:FN "Bob" . 
\end{Verbatim} 
\end{minipage}

\end{example}

\begin{example}
\label{ex:construct-two-B-run}
  
Consider the SPARQL CONSTRUCT query:

\begin{Verbatim}[frame=single,label=Query,fontsize=\scriptsize]
 CONSTRUCT { _:c vcard:FN ?name } WHERE { ?x foaf:name ?name }
\end{Verbatim}

\noindent
Note that this query always returns the same result as the query 
from Examples~\ref{ex:poim-one-B} and~\ref{ex:poim-one-B-run}.
Let us run this query against the RDF graph $G$ from
Example~\ref{ex:construct-two-V-run}.
There are two matches and we get the RDF graphs $H_1$, $H_2$
and the result $H$: 

\medskip
\noindent 
\begin{minipage}{1.3in}	
\begin{Verbatim}[frame=single,label=$H_1$,fontsize=\scriptsize] 
 _:c vcard:FN "Alice" . 
\end{Verbatim} 
\end{minipage}
\hfill 
\begin{minipage}{1.3in}	
\begin{Verbatim}[frame=single,label=$H_2$,fontsize=\scriptsize] 
 _:c vcard:FN "Bob" . 
\end{Verbatim} 
\end{minipage}
\hfill 
\begin{minipage}{1.3in}	
\begin{Verbatim}[frame=single,label=$H$,fontsize=\scriptsize] 
 _:c1 vcard:FN "Alice" . 
 _:c2 vcard:FN "Bob" . 
\end{Verbatim} 
\end{minipage}

\end{example}

Let $k$ be a natural number. 
According to Proposition~\ref{prop:graph-coprod}, for each query graph $T$
the query graph $k\,T$, coproduct of $k$ copies of $T$ in $\QGr{\I}$,
can be built (up to isomorphism) as follows:
for each $i\in\{1,...,k\}$ let $T_i$ be a copy of $T$
where each blank and variable has been renamed 
in such a way that there is no blank or variable common to two of the $T_i$'s, 
then the query graph $k\,T$ is the union $T_1 \cup ...\cup T_k$.
Now let $(L,R)$ be a basic construct query. 
The transformation rule $\rhoP_{L,R}=(L\upto{l} K\upfrom{r} R)$
is a cospan in $\QGr{\I}$, that gives rise to the cospan
$k\,\rhoP_{L,R}=(k\,L\upto{k\,l} k\,K\upfrom{k\,r} k\,R)$. 
Since $l$ and $r$ are inclusions, this renaming can be done simultaneously
in the copies of $L$, $K$ and $R$, so that $k\,K = k\,L \cup k\,R$
and $k\,l$ and $k\,r$ are the inclusions. 
Thus, $(k\,L,k\,R)$ is a basic construct query and 
$\rhoP_{k\,L,k\,R}=k\,\rhoP_{L,R}$ is its corresponding transformation rule. 

\begin{definition}
\label{def:construct-high}
Let $(L,R)$ be a basic construct query and $G$ a data graph.
Let $m_1,...,m_k$ be the matches from $L$ to $G$.
Consider the basic construct query $(k\,L,k\,R)$. 
Let $m$ be the match from $k\,L$ to $G$
that coincides with $m_i$ on the $i$-th component of $k\,L$.
The \emph{high-level query result of $(L,R)$ against $G$}
is the result $H_\high$ of applying the \poim transformation map 
$\fpoim_{k\,L,k\,R}$ to the match $m:k\,L\to G$,
as in Diagram~(\ref{diag:high}). 
\begin{equation}
\label{diag:high} 
\xymatrix@R=2pc@C=6pc{
  \ar@{}[rd]|{(PO)} k\,L \ar[r]^{k\,l} \ar[d]_{m} &
  \ar@{}[rd]|{(IM)} k\,K \ar@{.>}[d]^{n} &
  k\,R \ar[l]_{k\,r} \ar@{.>}[d]^{p} \\
  G \ar@{.>}[r]_{g} & D & H_\high \ar@{.>}[l]^{h} \\ 
} 
\end{equation}
\end{definition}

\begin{proposition}
\label{prop:construct-high} 
Let $(L,R)$ be a basic construct query and $G$ a data graph.
The high-level query result of $(L,R)$ against $G$ is isomorphic, 
in the category $\DGr{\I}$, to the query result of $(L,R)$ against $G$.
\end{proposition}

The low-level calculus is a two-step process: 
first one \emph{local result} is obtained for each match,
using a \poim transformation, 
then the local results are glued together in order to form the
\emph{low-level query result}. 

\begin{definition}
\label{def:construct-low}
Let $(L,R)$ be a basic construct query and $G$ a data graph. 
Let $m_1,...,m_k$ be the matches from $L$ to $G$.
For each $i=1,...,k$ let $G_i$ be the image of $m_i$
and let us still denote $m_i$ the restriction $m_i:L\to G_i$. 
The \emph{local result} $H_i$ of $(L,R)$ against $G$ along $m_i$ 
is the result of applying the \poim transformation map $\fpoim_{L,R}$
to the match $m_i:L\to G_i$.
Let $\IB(G)=\I\cup|G|_\B$. 
The \emph{low-level query result} $H_\low$ of $(L,R)$ against $G$ is the
coproduct of the $H_i$'s in the category $\DGr{\IB(G)}$ of data graphs with
morphisms fixing all resource identifiers and the blanks that are in $G$.  
\end{definition} 

\begin{equation}
\label{diag:low}
\xymatrix@R=2pc@C=6pc{
  \ar@{}[rd]|{(PO)} L \ar[r]^{l} \ar[d]_{m_i} &
  \ar@{}[rd]|{(IM)} K \ar@{.>}[d]^{n_i} &
  R \ar[l]_{r} \ar@{.>}[d]^{p_i} \\
  G_i \ar@{.>}[r]_{g_i} & D_i & H_i \ar@{.>}[l]^{h_i} \\ 
} 
\end{equation}

\begin{example}
\label{ex:construct-low} 
Let us apply the low-level calculus to Example~\ref{ex:construct-two-B-run}. 
The match $m_1$ produces $G_1 \to D_1 \from H_1$:

\medskip 
\noindent
\begin{minipage}{1.4in}
\begin{Verbatim}[frame=single,label=$G_1$,fontsize=\scriptsize] 
 ex:a foaf:name "Alice" .
\end{Verbatim} 
\end{minipage}
\hfill 
  $\upto{g_1}$
\hfill 
\begin{minipage}{1.4in}	
\begin{Verbatim}[frame=single,label=$D_1$,fontsize=\scriptsize] 
 ex:a foaf:name "Alice" .
 _:c vcard:FN "Alice" . 
\end{Verbatim}
\end{minipage}
\hfill 
  $\upfrom{h_1}$
\hfill 
\begin{minipage}{1.3in}	
\begin{Verbatim}[frame=single,label=$H_1$,fontsize=\scriptsize] 
 _:c vcard:FN "Alice" . 
\end{Verbatim} 
\end{minipage}

\noindent
and similarly the match $m_2$ produces $G_2 \to D_2 \from H_2$:

\medskip 
\noindent
\begin{minipage}{1.4in}
\begin{Verbatim}[frame=single,label=$G_2$,fontsize=\scriptsize] 
 ex:b foaf:name "Bob" .
\end{Verbatim} 
\end{minipage}
\hfill 
  $\upto{g_2}$
\hfill 
\begin{minipage}{1.4in}	
\begin{Verbatim}[frame=single,label=$D_2$,fontsize=\scriptsize] 
 ex:b foaf:name "Bob" . 
 _:c vcard:FN "Bob" . 
\end{Verbatim}
\end{minipage}
\hfill 
  $\upfrom{h_2}$
\hfill 
\begin{minipage}{1.3in}	
\begin{Verbatim}[frame=single,label=$H_2$,fontsize=\scriptsize] 
 _:c vcard:FN "Bob" . 
\end{Verbatim} 
\end{minipage}	

\smallskip
\noindent
Finally the query result $H_\low$, which is the coproduct of $H_1$ and $H_2$
in category $\DGr{\IB(G)}$, is isomorphic to $H$ from
Example~\ref{ex:construct-two-B-run}.

\begin{Verbatim}[frame=single,label=$H_\low$,fontsize=\scriptsize] 
 _:c1 vcard:FN "Alice" . _:c2 vcard:FN "Bob" . 
\end{Verbatim} 

\end{example}

\begin{example}
\label{ex:avecrecollement}

This example illustrates how local results are ``merged'' to compute
the result in the low-level calculus.
The SPARQL query is the following: 

\begin{Verbatim}[frame=single,label=Query,fontsize=\scriptsize]
CONSTRUCT { ?x rel:acquaintanceof ?z . } WHERE { ?x foaf:knows ?y . ?y foaf:knows ?z . }
\end{Verbatim}

\vspace{-5pt}
\noindent 
Its corresponding transformation rule is: 

\medskip 
\noindent
\begin{minipage}{1.1in}
\begin{Verbatim}[frame=single,label=$L$,fontsize=\scriptsize] 
 ?x foaf:knows ?y . 
 ?y foaf:knows ?z .   
\end{Verbatim} 
\end{minipage}
\hfill 
  $\upto{l}$
\hfill 
\begin{minipage}{1.5in}	
\begin{Verbatim}[frame=single,label=$K$,fontsize=\scriptsize] 
 ?x foaf:knows ?y ; 
    rel:acquaintanceOf ?z . 
 ?y foaf:knows ?z .        
\end{Verbatim}
\end{minipage}
\hfill 
  $\upfrom{r}$
\hfill 
\begin{minipage}{1.5in}	
\begin{Verbatim}[frame=single,label=$R$,fontsize=\scriptsize] 
 ?x rel:acquaintanceOf ?z . 
\end{Verbatim} 
\end{minipage}

\smallskip
\noindent 
This query is applied to the following graph $G$:

\begin{Verbatim}[frame=single,label=$G$,fontsize=\scriptsize] 
 <http://example.org/Alice> foaf:knows <http://example.org/Bob> .
 <http://example.org/Bob> foaf:knows _:c .
 _:c foaf:knows <http://example.org/Alice> .
\end{Verbatim}

\noindent 
There are three matches $m_1$, $m_2$, $m_3$,
thus three local results $H_1$, $H_2$, $H_3$: 

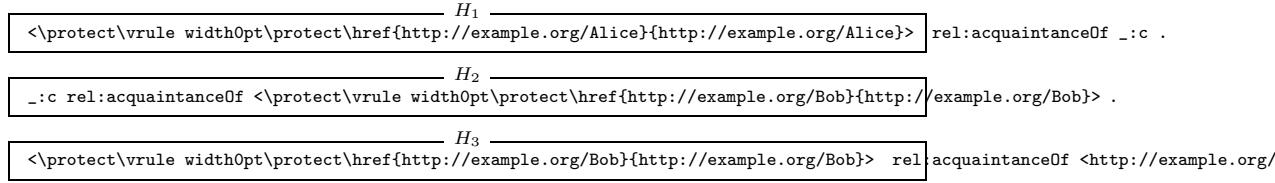
\begin{Verbatim}[frame=single,label=$H_1$,fontsize=\scriptsize]
 <http://example.org/Alice>  rel:acquaintanceOf _:c . 
\end{Verbatim}

\begin{Verbatim}[frame=single,label=$H_2$,fontsize=\scriptsize]
 _:c rel:acquaintanceOf <http://example.org/Bob> . 
\end{Verbatim}

\begin{Verbatim}[frame=single,label=$H_3$,fontsize=\scriptsize]
 <http://example.org/Bob>  rel:acquaintanceOf <http://example.org/Alice> .
\end{Verbatim}

\noindent 
The blank \mbox{{\tt \_:c}} in $H_1$ and $H_2$ is not duplicated in
the coproduct $H_\low$ because it comes from $G$. 
Thus the result is: 
	
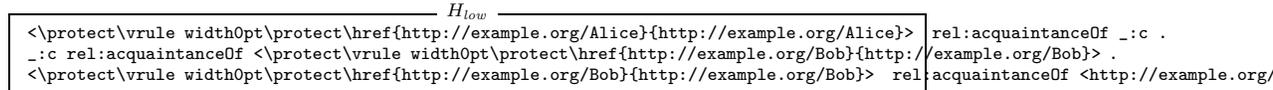
\begin{Verbatim}[frame=single,label=$H_\low$,fontsize=\scriptsize] 
 <http://example.org/Alice>  rel:acquaintanceOf _:c . 
 _:c rel:acquaintanceOf <http://example.org/Bob> . 
 <http://example.org/Bob>  rel:acquaintanceOf <http://example.org/Alice> .
\end{Verbatim}	
\end{example}

\begin{proposition}
\label{prop:construct-low}
Let $(L,R)$ be a basic construct query and $G$ a data graph.
The low-level query result of $(L,R)$ against $G$ is isomorphic, 
in the category $\DGr{\I}$, to the query result of $(L,R)$ against $G$.
\end{proposition}

%\input{5-select} 

%-------------------------------------------------------------------------------
%-------------------------------------------------------------------------------
\section{Running basic select queries} 
\label{sec:select}  

The CONSTRUCT query form of SPARQL returns a data graph
whereas the SELECT query form returns a table, 
like the SELECT query form of SQL.
Both in SQL and in SPARQL, it is well-known that such tables 
are not exactly relations in the mathematical sense:
in mathematics a relation on $X_1,...,X_n$ is a subset of the cartesian product
$X_1\!\times\!...\!\times\! X_n$, while the result of a SELECT query in SQL or SPARQL
is a multiset of elements of $X_1\!\times\!...\!\times\! X_n$. 
In order to avoid ambiguities, such a multiset 
is called a \emph{multirelation} on $X_1,...,X_n$.
When all $X_i$'s are the same set $X$ it is called a multirelation 
of arity $n$ on $X$.

A SPARQL query such as ``SELECT $?s_1,...,?s_n$ WHERE \{$L$\}'' is called
\emph{basic} when $L$ is a basic graph pattern and $?s_1,...,?s_n$
are distinct variables. 
We generalize this situation by defining a \emph{basic select query}
as a pair $(L,S)$ where $L$ is a finite query graph and $S$ is a finite set 
of variables. Then we associate to each basic select query $(L,S)$
a basic construct query $(L,\gr(S))$. Finally we define the result of
running the basic select query $(L,S)$ against a data graph
$G$ from the data graph $H$ result of running the basic construct query
$(L,\gr(S))$ against $G$.
This process is first described on an example. 

\begin{example} 
\label{ex:select-construct}
  
Consider the following SPARQL SELECT query:

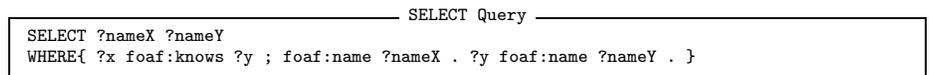
\begin{Verbatim}[frame=single,label=SELECT Query,fontsize=\scriptsize]
 SELECT ?nameX ?nameY
 WHERE{ ?x foaf:knows ?y ; foaf:name ?nameX . ?y foaf:name ?nameY . }
\end{Verbatim}

\noindent
We associate to this SELECT query the following CONSTRUCT query:

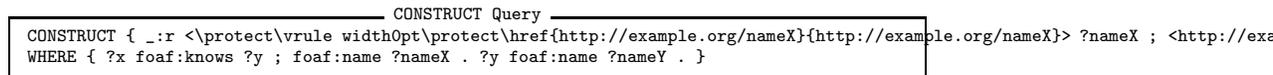
\begin{Verbatim}[frame=single,label=CONSTRUCT Query,fontsize=\scriptsize]
 CONSTRUCT { _:r <http://example.org/nameX> ?nameX ; <http://example.org/nameY> ?nameY . }
 WHERE { ?x foaf:knows ?y ; foaf:name ?nameX . ?y foaf:name ?nameY . }
\end{Verbatim}

\noindent
Let us run this  CONSTRUCT query against the RDF graph $G$: 

\begin{Verbatim}[frame=single,label=$G$,fontsize=\scriptsize] 
 _:a foaf:name "Alice" ; foaf:knows _:b ; foaf:knows _:c .
 _:b foaf:name "Bob" . 
 _:c foaf:name "Cathy" . 
\end{Verbatim}

\noindent
The result is the RDF graph $H$: 

\begin{Verbatim}[frame=single,label=$H$,fontsize=\scriptsize] 
 _:l1 <http://example.org/nameX> "Alice" ; <http://example.org/nameY> "Bob" .
 _:l2 <http://example.org/nameX> "Alice" ; <http://example.org/nameY> "Bob" .
 _:l3 <http://example.org/nameX> "Alice" ; <http://example.org/nameY> "Cathy" .
\end{Verbatim}

\noindent
{F}rom the RDF graph $H$ we get the following table, 
by considering each blank \blank{$l_i$} in $H$
as the identifier of a line in the table. Note that the set of triples
in $H$ becomes a multiset of lines in the table.
This table is indeed the answer of the SPARQL SELECT query over $G$.

\begin{center}
  \scriptsize
  \begin{tabular}{|l|l|}
  \hline
\multicolumn{1}{|c|}{\;\textbf{nameX}\;} & 	\multicolumn{1}{|c|}{\;\textbf{nameY}\;} \\	
  \hline
\mbox{{\tt "Alice"}} & \mbox{{\tt "Bob"}} \\ 
  \hline
\mbox{{\tt "Alice"}} & \mbox{{\tt "Bob"}} \\ 
   \hline
\mbox{{\tt "Alice"}} & \mbox{{\tt "Cathy"}} \\ 
  \hline
  \end{tabular}
  \normalsize
\end{center}

\end{example}

In order to generalize Example~\ref{ex:select-construct} we have to
define a transformation from each SELECT query to a CONSTRUCT query 
and a transformation from the result of this CONSTRUCT query
to the result of the given SELECT query. 
For this purpose,
we first define \emph{relational} data graphs (Definition~\ref{def:select-data})
and \emph{relational} query graphs (Definition~\ref{def:select-query}). 

\begin{definition}
\label{def:select-data}
A \emph{relational data graph} on a finite set $\{s_1,...,s_n\}$ of
resource identifiers 
is a data graph made of triples $(\_:l_i,s_j,y_{i,j})$ 
where the $\_:l_i$'s are pairwise distinct blanks 
and the $y_{i,j}$'s are in $\IB$, 
for $j\in\{1,...,n\}$ and $i$ in some finite set $\{1,...,k\}$.
\end{definition}

\begin{proposition}
\label{prop:select-relation}
\ \\
Each relational data graph $S=\{(\_:l_i,s_j,y_{i,j})\}_{i\in\{1,...,k\},j\in\{1,...,n\}}$
determines a multirelation $\rel(S)=\{(y_{i,1},...,y_{i,n})\}_{i\in\{1,...,k\}}$
of arity $n$ on $\IB$.
\end{proposition}

\begin{example}
\label{ex:select-rel-data}
Here is a relational data graph on 
$\{\mbox{{\tt nameX}},\mbox{{\tt nameY}}\}$
with its corresponding multirelation: 

\smallskip
\noindent\hfill
\begin{minipage}{2in}
\begin{Verbatim}[frame=single,fontsize=\scriptsize] 
 _:l1 nameX "Alice" ; nameY "Bob" .
 _:l2 nameX "Alice" ; nameY "Cathy" .
 _:l3 nameX "Alice" ; nameY "Cathy" .
\end{Verbatim}
\end{minipage}
\hfill 
\scriptsize
\begin{tabular}{|l|l|}
  \hline
\multicolumn{1}{|c|}{\;\textbf{nameX}\;} & \multicolumn{1}{|c|}{\;\textbf{nameY}\;} \\	
  \hline
\mbox{{\tt "Alice"}} & \mbox{{\tt "Bob"}} \\ 
  \hline
\mbox{{\tt "Alice"}} & \mbox{{\tt "Cathy"}} \\ 
  \hline
\mbox{{\tt "Alice"}} & \mbox{{\tt "Cathy"}} \\ 
  \hline
\end{tabular}
\normalsize
\hfill \null

\end{example}

Assume that each variable in SPARQL is written as ``$?s$''
for some string $s$. 

\begin{definition}
\label{def:select-query}
The \emph{relational query graph} on a finite set of variables
$S=\{?s_1,...,?s_n\}$
is the query graph $\gr(S)$ made of the triples $(\_:r,s_j,?s_j)$
where $j\in \{1,...,n\}$ and $\_:r$ is a blank.
Note that $\gr(S)$ is uniquely determined by $S$
up to isomorphism in $\QGr{\IV}$.
\end{definition}

\begin{example}
\label{ex:select-rel-query}
Here is the relational query graph on
$\{\mbox{{\tt ?nameX}},\mbox{{\tt ?nameY}}\}$:

\begin{center}
\begin{minipage}{2in}
\begin{Verbatim}[frame=single,fontsize=\scriptsize] 
 _:r nameX ?nameX ; nameY ?nameY .
\end{Verbatim}
\end{minipage}
\end{center}

\end{example}

In the following we show how a basic select query can be encoded as a 
basic construct query (Definition~\ref{def:select-select})
and we prove that the result of the given select query is easily recovered
from the result of its associated construct query (Theorem~\ref{theo:select-sparql}).

\begin{definition}
\label{def:select-select}
A \emph{basic select query} is a pair $(L,S)$ where $L$ is a finite
query graph and $S$ is a finite set of variables 
such that each variable in $S$ occurs in $L$.
The \emph{basic construct query associated to} a basic select query $(L,S)$
is $(L,\gr(S))$ where $\gr(S)$ is the relational query graph on $S$. 
\end{definition}

\begin{proposition} 
\label{prop:select-result}
Let $(L,S)$ be a basic select query and $G$ a data graph.
The query result of $(L,\gr(S))$ against $G$ is a relational data graph $H$.
More precisely, let $S=\{?s_1,...,?s_n\}$ and let $m_1,...,m_k$ be
the matches from $L$ to $G$, then $H$ is the set of triples
$(\_:l_i,s_j,m_i(?s_j))$ where $i\in \{1,...,k\}$, $j\in \{1,...,n\}$,
and the blanks $\_:l_1,...,\_:l_k$ are pairwise distinct. 
\end{proposition}

Because of Proposition~\ref{prop:select-result} we can state the
following definition. 

\begin{definition}
\label{def:select-result}
Let $(L,S)$ be a basic select query and $G$ a data graph.
Let $H$ be the query result of $(L,\gr(S))$ against $G$. 
The \emph{query result} of $(L,S)$ against $G$
is the multirelation $\rel(H)$ on $\IB$. 
\end{definition}
  
\begin{theorem}
\label{theo:select-sparql}
Let $L$ be a basic graph pattern of SPARQL
and $S=\{?s_1,...,?s_n\}$
a finite set of variables included in $|L|_\V$ and let $G$ be an RDF graph. 
Then the query result of $(L,R)$ against $G$ is the answer
of the SPARQL query ``SELECT $?s_1,...,?s_n$ WHERE \{$L$\}'' over $G$.  
\end{theorem}

\begin{remark}
\label{rem:select-null}
This Section can be generalized to multirelations with ``null'' values
(as in SQL) 
by allowing some missing triples in the definition of relational data graphs. 
\end{remark}

%\input{6-discuss} 

%-------------------------------------------------------------------------------
%-------------------------------------------------------------------------------
\section{Conclusion} 
\label{sec:discuss}

Relational algebra \cite{Codd_90} is the main mathematical
foundation underlying SQL-like formalism for databases.
However new frameworks such as RDF and SPARQL, where data structures are
represented as graphs, are better adapted 
to the needs of big data and web applications.
So, new mathematical
foundations are needed to cope with this change in data encodings, see
e.g., \cite{AnglesABHRV17,PerezAG09,KKC}.

In this paper, we make the bet to base our work entirely on algebraic
theories behind graphs and their transformations.  Suitable categories
of data graphs and query graphs are defined and the definition of
morphisms of query graphs clarifies the difference between blank nodes
and variables. Besides, we propose to encode CONSTRUCT and SELECT
queries as graph rewrite rules, of the form
$L \rightarrow L \cup R \leftarrow R$, and define their operational
semantics following a novel algebraic approach called POIM.  From the
proposed semantics, blanks in $L$ play the same role as variables and
thus can be replaced by variables, whereas blanks in $R$ are used for
creating new blanks in the result of a CONSTRUCT query.  As in
\cite{Kostylev_et_al2015}, we focus on the CONSTRUCT query form as the
fundamental query form.  In addition we propose a translation of the
SELECT queries as CONSTRUCT queries compatible with their operational
semantics. One of the benefits of using category theory is that
coding of data graphs as sets of triples is not that important. The
results we propose hold for all data models which define a category
with enough colimits. For intance, one may expect to define data graph
categories for the well-known Edge-labelled graphs or Property graphs
\cite{Robinson2013}.
The proposed operational semantics can clearly benefit from all
results regarding efficient graph matching implementation, see
e.g. \cite{FanLMWW10}. 

Among related works, a category of RDF-graphs as well as their
transformations have been proposed in \cite{BraatzB08}. The authors
defined  objects
of  RDF-Graph categories  of the form
$(G_{\mathrm{Blank}}, G_{\mathrm{Triple}})$ where $G_{\mathrm{Blank}}$
and $G_{\mathrm{Triple}}$ denote respectively the set of blank nodes
and the set of triples of graph $G$. This definition is clearly
different from ours (Definition~\ref{def:graph-data-query}). In
addition, the morphisms of such RDF-graphs associate blank nodes to
blank nodes which is not always the case in our approach. Associating
a blank node to any element of a triple is called instantiation in
\cite{BraatzB08}. The authors did not tackle the problem of answering
SPARQL queries but rather proposed an algebraic approach to transform
RDF-graphs. Their approach, called MPOC-PO, is inspired from DPO where
the first square is replaced by a ``minimal'' pushout complement
(MPOC). MPOC-PO drastically departs from the POIM transformations we
propose. This difference is quite natural since the two approaches
have different objectives~: the POIM approach is dedicated to
implement SPARQL queries while the MPOC-PO is intended to transform
RDF-graphs in general. However, MPOC-PO and DPO approaches are clearly
not tailored to implement CONSTRUCT or SELECT queries since the
(minimal) pushout complements always include parts of large data
graphs which are not matched by the queries while such parts are not
involved in the query answers.

In \cite{AJD2015}, even if the authors use a categorical setting,
their objectives and results depart from ours as they mainly encode
every ontology as a category. However, Graph Transformations have
already been used in modeling relational databases, see
e.g. \cite{AndriesE96} where a visual and textual hybrid query
language has been proposed. In \cite{KieselSW95}, the main features of
a data management system based on graphs have been proposed where the
underlying typed attributed data graphs are different from those of
RDF and SPARQL. In \cite{AlqahtaniH18}, triple graph grammars (TGG)
have also been used for data modelling and model transformation rules
to be compiled into Graph Data Bases code for execution.

In this paper we consider basic graphs and queries, which form 
a significant kernel of RDF and SPARQL. 
Future work includes the generalization of the present work to other 
features of RDF and SPARQL in order to encompass general SPARQL queries.
We also consider studying RDF Schema \cite{rdfs} and ontologies
from this point of view.

%\input{7-appendix}

%-------------------------------------------------------------------------------
%-------------------------------------------------------------------------------
\appendix
\section{Proofs}
\label{sec:appendix-proofs}

%-------------------------------------------------------------------------------
First let us prove the results about colimits and \poim
transformations in the category $\tgr{C}{A}$.
It can be helpful to remember that a morphism $a:T\to T'$ in $\tgr{}{A}$
and a map $M:|T|\to|T'|$ are such that $M(x)=|a|(x)$ for each attribute
$x\in |T|$ if and only if $M^3(t)=a(t)$ for each triple $t\in T$. 

\setcounterref{theorem}{prop:graph-coprod}
\addtocounter{theorem}{-1}
\begin{proposition}  
Given graphs $T_1,...,T_k$ on $A$
such that $|T_i|\cap|T_j|\subseteq C$ for each $i\ne j$,
the union $T_1 \cup ...\cup T_k$ is a coproduct of $T_1,...,T_k$
in $\tgr{C}{A}$.
\end{proposition}

\begin{proof}
Consider morphisms $a_i:T_i\to T$ in $\tgr{C}{A}$ for $i=1,...,k$
and the maps $|a_i|:|T_i|\to |T|$. Note that $|T_1 \cup ...\cup T_k| =
|T_1| \cup ...\cup |T_k|$ and that $|T_1| \cup ...\cup |T_k|$
is the disjoint union of the sets $|T_i|\!\setminus\! C$ for $i=1,...,k$ and
$(|T_1| \cup ...\cup |T_k|)\cap C$, because of the assumption
$|T_i|\cap|T_j|\subseteq C$ for each $i\ne j$.
Thus we can define a map $M:|T_1 \cup ...\cup T_k|\to |T|$ by: 
$M(x)=|a_i|(x)$ for each $i$ and each $x\in |T_i|\!\setminus\! C$
and $M(x)=x$ for each $x\in (|T_1| \cup ...\cup |T_k|)\cap C$.
Then $M$ coincides with $|a_i|$ on $|T_i|$ for each $i$.
Thus for each $t\in T_i$ we have $M^3(t)=a_i(t)$,
which proves that the image of $T_1 \cup ...\cup T_k$ by $M^3$ is in $T$
and that the restriction of $M^3$ defines a morphism
$a:T_1 \cup ...\cup T_k\to T$ in $\tgr{C}{A}$ which
coincides with $a_i$ on $T_i$ for each~$i$.
Unicity is clear. 
\end{proof} 

\setcounterref{theorem}{prop:graph-po}
\addtocounter{theorem}{-1}
\begin{proposition}  
Let $l:L\to K$ and $m:L\to G$ be morphisms of graphs on $A$
such that $K$ is finite, $l$ is an inclusion and $m$ fixes $C$. 
Let us assume that $|G| \cap |K| \subseteq C$
(this is always possible up to isomorphism in $\tgr{C}{A}$,
by Remark~\ref{rem:graph-iso}).
Let $N:|K|\to A$ be such that $N(x)=|m|(x)$ for $x\in |L|$ 
and $N(x)=x$ otherwise. 
Let $D=G\cup N^3(K)$, let $n:K\to D$ be the restriction of $N^3$
and $g:G\to D$ the inclusion.
Then $|D|=|G|\cup |K\!\setminus\! L|$ and
the square $(l,m,n,g)$ is a pushout square in $\tgr{C}{A}$.
\end{proposition}

\begin{proof}
{F}rom $D=G\cup N^3(K)$ we get $|D|=|G|\cup |N^3(K)|$, and since
$|N^3(K)|=N(|K|)=N(|L|\cup|K\!\setminus\! L|)=N(|L|)\cup N(|K\!\setminus\! L|)
=|m|(|L|)\cup |K\!\setminus\! L|$ with $|m|(|L|)\subseteq |G|$
we get $|D|=|G|\cup |K\!\setminus\! L|$. 
The definition of $n$ implies that $g\circ m = n\circ l$. 
Now let $a:G\to T$ and $b:K\to T$ be any morphisms in $\tgr{C}{A}$ such that
$a\circ m = b\circ l$. 
First, let us focus on attributes.
We have $|g|\circ |m| = |n|\circ |l|$ and $|a|\circ |m| = |b|\circ |l|$. 
Since $|G| \cap |K\!\setminus\! L| \subseteq C$ we have $|a|(x)=|b|(x)=x$
for each $x\in |G| \cap |K\!\setminus\! L|$. 
Since $|D|=|G|\cup |K\!\setminus\! L|$ there is a unique map 
$F:|D|\to |T|$ such that $F(x)=|a|(x)$ for $x\in |G|$ and
$F(x)=|b|(x)$ for $x\in |K\!\setminus\! L|$.
Thus on the one hand $F(|g|(x)) = F(x) =|a|(x)$ for each $x\in |G|$, so that
$F\circ |g|=|a|$.
And on the other hand for each $x\in |K|$, if $x\in |L|$ then 
$F(|n|(x)) = F(|m|(x)) = |a|(|m|(x)) = |b|(|l|(x))=|b|(x)$,
otherwise $F(|n|(x)) = F(x) = |b|(x)$, so that $F\circ |n|=|b|$. 
Second, let us consider triples. 
Since $D=G\cup N^3(K)$ and $F^3(G) =a(G) $ and $F^3(N^3(K)) = F^3(n(K))= b(K)$
we get $F^3(D) \subseteq T$, which means that there is a morphism
$f:D\to T$ of graphs on $A$ such that $|f|=F$, $f\circ g=a$ and $f\circ n=b$.
Unicity is clear. 
\end{proof}

\setcounterref{theorem}{prop:construct-poim}
\addtocounter{theorem}{-1}
\begin{proposition}
Let $(L,R)$ be a basic construct query and $m:L\to G$ a match. 
Let $P:|R|\to A$ be defined by $P(x)=|m|(x)$ for $x\in|R|_\V$ and $P(x)=x$
otherwise.
Then, up to isomorphism in $\QGr{\I}$, 
the result of applying $\fpoim_{L,R}$ to $m$ is $p:R\to H$
where $H=P^3(R)$ and $p$ is the restriction of $P^3$. 
\end{proposition}

\begin{proof}
We use the notations of Diagram~(\ref{diag:poim}).
Up to isomorphism in $\QGr{\I}$ we can assume that all blanks in $L$ or in $R$
are distinct from the blanks in $G$. Then $|G| \cap |K| \subseteq C$,
so that by Proposition~\ref{prop:graph-po} 
the data graph $D$ is $D=G\cup n(K)$ where $n$ is such that
$|n|(x)=|m|(x)$ for $x\in |L|$ and $|n|(x)=x$ otherwise.
It follows that the restriction of $n$ to $R$ is such that
$|n|(x)=|m|(x)$ for $x\in |L|\cap|R|$ and $|n|(x)=x$ otherwise.
Note that $|L|\cap |R|$ is the disjoint union of
$|L|_\I\cap|R|_\I$, that is fixed by all morphisms in $\QGr{\I}$, 
and $|L|_\V\cap|R|_\V$, with $|L|_\V\cap|R|_\V=|R|_\V$
since $|R|_\V\subseteq|L|_\V$.
Thus the restriction of $n$ to $R$ is such that
$|n|(x)=|m|(x)$ for $x\in |R|_\V$ and $|n|(x)=x$ otherwise.
The result follows. 
\end{proof}

%-------------------------------------------------------------------------------
Now let us consider the basic construct queries.
The semantics of SPARQL CONSTRUCT queries is defined in 
\cite[Section~5]{Kostylev_et_al2015},
based on the seminal paper \cite{PerezAG09}. 
In order to express this definition we have to introduce some terminology and
notations. 
Note that in \cite{Kostylev_et_al2015} 
literals are allowed as subjects or predicates in RDF graphs.
However for our purpose this does not matter, so that
we stick to the ``official'' definition of an RDF graph from \cite{rdf}. 
Note that for each subset $T$ of $(\IBV)^3$ and each subset $X$ of $|T|$, 
each map $f:X\to \IBV$ gives rise to a map $f':|T|\to \IBV$
such that $f'(x)=f(x)$ when $x\in X$ and $f'(x)=x$ otherwise, 
then $f':|T|\to \IBV$ gives rise to $f'':T\to (\IBV)^3$ which is the
restriction of $(f')^3$ to $T$. There will not be any ambiguity in denoting
$f$ not only the given $f$ but also its extensions $f'$ and $f''$,
so that we can state the following definitions from
\cite{Kostylev_et_al2015}. 
For simplicity we consider only the SPARQL queries ``CONSTRUCT $\{R\}$
WHERE $\{L\}$'' such that each variable in $R$ occurs in $L$.
Indeed, variables outside $|L|_\V$ cannot be instanciated in the result,
and according to \cite[Section~16.2]{sparql}, 
if a triple contains an unbound variable,
then that triple is not included in the output RDF graph.
Thus, triples involving a variable in $|R|_\V\!\setminus\!|L|_\V$,
if any, can be dropped. 
It is assumed in \cite{Kostylev_et_al2015} that there is no blank in $L$.
Indeed, since blank nodes in graph patterns act as variables,
each blank in $L$ can be replaced by a new variable.
A \emph{solution mapping} (or simply a \emph{mapping}) from a basic graph
pattern $L$ to an RDF graph $G$ is a map $\mu:|L|_\V\to \IB$ such that
$\mu(L)\subseteq G$.
When $L$ and $R$ are basic graph patterns such that $|R|_\V\subseteq|L|_\V$,
the \emph{answer} of the SPARQL query ``CONSTRUCT $\{R\}$ WHERE $\{L\}$''
over an RDF graph $G$ 
is the set of all well-formed triples $\mu(f_\mu(t))$ 
for all triples $t\in R$ and all mappings $\mu$ from $L$ to $G$, 
where for each $\mu$ a map $f_\mu:|R|_\B\to \B$ is chosen in such a way that the
subsets $f_\mu(|R|_\B)$ of $\B$ are pairwise distinct and all of them are
distinct from $|G|_\B$.

\setcounterref{theorem}{theo:construct-sparql}
\addtocounter{theorem}{-1}
\begin{theorem}
Let $L$ and $R$ be basic graph patterns with $|L|_\B=\emptyset$
and $|R|_\V\subseteq |L|_\V$.
Then $(L,R)$ is a basic construct query and 
the set of well-formed triples in the query result of applying $(L,R)$
to an RDF graph $G$ is isomorphic in $\DGr{\I}$ to the answer
of the SPARQL query ``CONSTRUCT $\{R\}$ WHERE $\{L\}$'' over $G$.  
\end{theorem}

\begin{proof}
Clearly $(L,R)$ is a basic construct query and $|G|_\B\cap |L|_\B=\emptyset$. 
We can assume without loss of generality that $|G|_\B\cap |R|_\B=\emptyset$. 
The query result $H$ of applying $(L,R)$ to $G$ is given by 
Definition~\ref{def:construct-result}, as reminded now. 
Let $m_i$ ($i=1,...,k$) be the matches from $L$ to $G$ and 
for each $i$ let $H_i$ be the data graph obtained from $R$
by replacing each variable $x$ in $R$ by $m_i(x)$
and each blank in $R$ by a new blank, then $H=H_1 \cup ...\cup H_k$.
The Theorem now follows from the remark that the maps $\mu''$ on triples
which are associated (as above) to the mappings $\mu$ are precisely
the matches from $L$ to $G$.
% suffisant ? ***
\end{proof}

\setcounterref{theorem}{prop:construct-high}
\addtocounter{theorem}{-1}
\begin{proposition}
Let $(L,R)$ be a basic construct query and $G$ a data graph.
The high-level query result of $(L,R)$ against $G$ is isomorphic, 
in the category $\DGr{\I}$, to the query result of $(L,R)$ against $G$.
\end{proposition}

\begin{proof}
This is a consequence of the description of the result of a
\poim transformation from Proposition~\ref{prop:construct-poim}.
% suffisant ? ***
\end{proof}

\setcounterref{theorem}{prop:construct-low} 
\addtocounter{theorem}{-1}
\begin{proposition}
Let $(L,R)$ be a basic construct query and $G$ a data graph.
The low-level query result of $(L,R)$ against $G$ is isomorphic, 
in the category $\DGr{\I}$, to the query result of $(L,R)$ against $G$.
\end{proposition}

\begin{proof}
This is a consequence of the description of the result of a
\poim transformation from Proposition~\ref{prop:construct-poim}
and the description of coproducts in $\DGr{\I}$
from Proposition~\ref{prop:graph-coprod}.
% suffisant ? ***
\end{proof}

%-------------------------------------------------------------------------------
Finally let us look at the basic select queries.
For the semantics of SPARQL SELECT queries we rely on \cite[Section~2]{KKC}. 

\setcounterref{theorem}{prop:select-relation}
\addtocounter{theorem}{-1}
\begin{proposition}
Each relational data graph $S=\{(\_:l_i,s_j,y_{i,j})\}_{i\in\{1,...,k\},j\in\{1,...,n\}}$
determines a multirelation $\rel(S)=\{(y_{i,1},...,y_{i,n})\}_{i\in\{1,...,k\}}$
of arity $n$ on $\IB$.
\end{proposition}

\begin{proof}
This result is clear from the definitions of relational data graphs 
and multirelations. 
\end{proof}

\setcounterref{theorem}{prop:select-result}
\addtocounter{theorem}{-1}
\begin{proposition} 
Let $(L,S)$ be a basic select query and $G$ a data graph.
The query result of $(L,\gr(S))$ against $G$ is a relational data graph $H$.
More precisely, let $S=\{?s_1,...,?s_n\}$ and let $m_1,...,m_k$ be
the matches from $L$ to $G$, then $H$ is the set of triples
$(\_:l_i,s_j,m_i(?s_j))$ where $i\in \{1,...,k\}$, $j\in \{1,...,n\}$,
and the blanks $\_:l_1,...,\_:l_k$ are pairwise distinct. 
\end{proposition}

\begin{proof}
We have $\gr(S)=\{(\_:r,s_j,?s_j)\}_{j\in \{1,...,n\}}$, so that 
according to Definition~\ref{def:construct-result}  
the query result of $(L,\gr(S))$ against $G$ is $H_1 \cup ...\cup H_k$
where $H_i=\{(\_:l_i,s_j,m_i(?s_j))\}_{j\in \{1,...,n\}}$
and the blanks $\_:l_1,...,\_:l_k$ are pairwise distinct.
\end{proof}

\setcounterref{theorem}{theo:select-sparql}
\addtocounter{theorem}{-1}
\begin{theorem}
Let $L$ be a basic graph pattern of SPARQL with $|L|_\B=\emptyset$ 
and $S=\{?s_1,...,?s_n\}$
a finite set of variables included in $|L|_\V$ and let $G$ be an RDF graph. 
Then the query result of $(L,R)$ against $G$ is the answer
of the SPARQL query ``SELECT $?s_1,...,?s_n$ WHERE \{$L$\}'' over $G$.  
\end{theorem}

\begin{proof}
According to \cite[Section~2]{KKC}, the answer 
of the SPARQL SELECT query is the multiset with elements the restrictions
$\mu|_S$ of the mappings $\mu$ from $L$ to $G$ to the subset $S$ of $|L|_\V$,
each $\mu|_S$ with multiplicity the number of corresponding $\mu$'s.
Since the mappings from $L$ to $G$ correspond bijectively to the matches
from $L$ to $G$, the result follows from Proposition~\ref{prop:select-result}.
\end{proof}


\begin{thebibliography}{10}
\providecommand{\url}[1]{\texttt{#1}}
\providecommand{\urlprefix}{URL }

\bibitem{AJD2015}
Aliyu, S., Junaidu, S., Kana, A.F.D.: A category theoretic model of {RDF}
  ontology. International Journal of Web \& Semantic Technology (IJWesT)
  (2015)

\bibitem{AlqahtaniH18}
Alqahtani, A., Heckel, R.: Model based development of data integration in graph
  databases using triple graph grammars. In Software Technologies: Applications and Foundations. Lecture Notes in Computer Science, vol. 11176, pp. 399--414.
  Springer (2018), \url{https://doi.org/10.1007/978-3-030-04771-9\_29}

\bibitem{AndriesE96}
Andries, M., Engels, G.: A hybrid query language for an extended
  entity-relationship model. J. Vis. Lang. Comput.  7(3),  321--352 (1996),
  \url{https://doi.org/10.1006/jvlc.1996.0017}

\bibitem{AnglesABHRV17}
Angles, R., Arenas, M., Barcel{\'{o}}, P., Hogan, A., Reutter, J.L., Vrgoc, D.:
  Foundations of modern query languages for graph databases. {ACM} Comput.
  Surv.  50(5),  68:1--68:40 (2017), \url{https://doi.org/10.1145/3104031}

\bibitem{BraatzB08}
Braatz, B., Brandt, C.: Graph transformations for the resource description
  framework. {ECEASST}  10 (2008),
  \url{https://doi.org/10.14279/tuj.eceasst.10.158}

\bibitem{Codd_90}
Codd, E.F.: The relational Model for Database Management (Version 2 ed.).
  Addison-Wesley (1990)

\bibitem{CorradiniHHK06}
Corradini, A., Heindel, T., Hermann, F., K{\"{o}}nig, B.: Sesqui-pushout
  rewriting. In: {ICGT} 2006. LNCS, vol. 4178, pp. 30--45. Springer (2006)

\bibitem{CorradiniMREHL97}
Corradini, A., Montanari, U., Rossi, F., Ehrig, H., Heckel, R., L{\"{o}}we, M.:
  Algebraic approaches to graph transformation - part {I:} basic concepts and
  double pushout approach. In: Rozenberg, G. (ed.) Handbook of Graph Grammars.
  pp. 163--246. World Scientific (1997)

\bibitem{DBLP:conf/gg/EhrigHKLRWC97}
Ehrig, H., Heckel, R., Korff, M., L{\"{o}}we, M., Ribeiro, L., Wagner, A.,
  Corradini, A.: Algebraic approaches to graph transformation - part {II:}
  single pushout approach and comparison with double pushout approach. In:
  Rozenberg, G. (ed.) Handbook of Graph Grammars and Computing by Graph
  Transformations, Volume 1: Foundations, pp. 247--312. World Scientific (1997)

\bibitem{FanLMWW10}
Fan, W., Li, J., Ma, S., Wang, H., Wu, Y.: Graph homomorphism revisited for
  graph matching. {PVLDB}  3(1),  1161--1172 (2010),
  \url{http://www.vldb.org/pvldb/vldb2010/pvldb\_vol3/R103.pdf}

\bibitem{KKC} Kaminski, M., Kostylev, E.V., Grau, B.C.: Semantics and
  expressive power of subqueries and aggregates in {SPARQL} 1.1. In:
  Proceedings of the 25th International Conference on World Wide Web,
  {WWW} 2016. pp. 227--238. {ACM} (2016)

\bibitem{KieselSW95}
Kiesel, N., Sch{\"{u}}rr, A., Westfechtel, B.: Gras, a graph-oriented
  (software) engineering database system. Inf. Syst.  20(1),  21--51 (1995),
  \url{https://doi.org/10.1016/0306-4379(95)00002-L}

\bibitem{Kostylev_et_al2015}
Kostylev, E.V., Reutter, J.L., Ugarte, M.: {CONSTRUCT} queries in {SPARQL}. In:
  18th International Conference on Database Theory, {ICDT} 2015, March 23-27,
  2015, Brussels, Belgium. pp. 212--229 (2015)

\bibitem{PerezAG09}
P{\'{e}}rez, J., Arenas, M., Guti{\'{e}}rrez, C.: Semantics and complexity of
  {SPARQL}. {ACM} Trans. Database Syst.  34(3),  16:1--16:45 (2009),
  \url{https://doi.org/10.1145/1567274.1567278}

\bibitem{Robinson2013}
Robinson, I., Webber, J., Eifrem, E.: Graph Databases. O'Reilly Media, Inc.
  (2013)

\bibitem{sparql}
{SPARQL 1.1 Query Language}. W3C Recommendation (march 2013),
  \url{https://www.w3.org/TR/sparql11-query/}

\bibitem{rdf}
{RDF 1.1 Concepts and Abstract Syntax}. W3C Recommendation (February 2014),
  \url{https://www.w3.org/TR/rdf11-concepts/}

\bibitem{rdfs}
{RDF Schema 1.1}. W3C Recommendation (February 2014),
  \url{www.w3.org/TR/2014/REC-rdf-schema-20140225/},
  \url{www.w3.org/TR/2014/REC-rdf-schema-20140225/}

\end{thebibliography}
\end{document}